\documentclass[runningheads]{llncs}
\usepackage[T1]{fontenc}
\usepackage{graphicx}
\usepackage{algorithm}
\usepackage[noend]{algpseudocode}
\usepackage{amsmath,amssymb}
\newtheorem{observation}{Observation}
\usepackage{paralist}
\usepackage{cleveref}
\crefname{observation}{Observation}{Observations}
\Crefname{observation}{Observation}{Observations}

\title{Tight Bounds for Constant-Round Domination on Graphs of High Girth and Low Expansion}
\titlerunning{Constant-Round Domination with High Girth and Low Expansion}

\begin{document}
\author{Christoph Lenzen and Sophie Wenning}
\authorrunning{C. Lenzen and S. Wenning}
\institute{CISPA Helmholtz Center for Information Security \email{\{lenzen,sophie.wenning\}@cispa.de}}
\date{}
\maketitle

\begin{abstract}
A long-standing open question is which graph class is the most general one permitting constant-time constant-factor approximations for dominating sets.
The approximation ratio has been bounded by increasingly general parameters such as genus, arboricity, or expansion of the input graph.
Amiri and Wiederhake considered $k$-hop domination in graphs of bounded $k$-hop expansion and girth at least $4k+3$~\cite{amiri22distributed};
the $k$-hop expansion $f(k)$ of a graph family denotes the maximum ratio of edges to nodes that can be achieved by contracting disjoint subgraphs of radius $k$ and deleting nodes.
In this setting, these authors to obtain a simple $O(k)$-round algorithm achieving approximation ratio $\Theta(kf(k))$.

In this work, we study the same setting but derive tight bounds:
\begin{itemize}
  \item A $\Theta(kf(k))$-approximation is possible in $k$, but not $k-1$ rounds.
  \item In $3k$ rounds an $O(k+f(k)^{k/(k+1)})$-approximation can be achieved.
  \item No constant-round deterministic algorithm can achieve approximation ratio $o(k+f(k)^{k/(k+1)})$.
\end{itemize}
Our upper bounds hold in the port numbering model with small messages, while the lower bounds apply to local algorithms, i.e., with arbitrary message size and unique identifiers.
This means that the constant-time approximation ratio can be \emph{sublinear} in the edge density of the graph, in a graph class which does not allow a constant approximation.
This begs the question whether this is an artefact of the restriction to high girth or can be extended to all graphs of $k$-hop expansion $f(k)$.
\end{abstract}

\section{Introduction}

Given a graph $G=(V,E)$, a $k$-hop Minimum Dominating Set (MDS) $M\subseteq V$ minimizes $|M|$ under the constraint that all nodes are within distance at most $k$ of a node in $M$.
The classic case of $k=1$ has been well-studied in the distributed setting, resulting in constant-time constant-factor approximations for a variety of sparse graph families:
\begin{inparaenum}[(i)]
  \item outerplanar~\cite{bonamy21tight},
  \item planar~\cite{hilke14local,lenzen13distributed,wawrzyniak14strengthened,wawrzyniak15local},
  \item bounded genus~\cite{amiri19distributed,czygrinow19distributed},
  \item excluded $K_t$ minor~\cite{czygrinow18distributed},
  \item excluded $K_{2,t}$ minor~\cite{czygrinow22distributed},\footnote{This work addresses $k$-hop domination. Moreover, the authors point out that the result can be generalized to require only exclusion of shallow $K_{2,t}$-minors.} and
  \item bounded $1$-hop expansion\footnote{A graph has $r$-hop expansion $f(r)$, if deleting nodes and contracting disjoint subgraphs of radius at most $r$ results in ratio of at most $f(r)$ between edges and nodes. Contracting a subgraph means to replace it by a single node sharing an each with each other node that was adjacent to a node of the subgraph.}~\cite{kublenz21constant}.
\end{inparaenum}
All of these algorithms share the property that their approximation ratio is at least linear in the edge density of the input graph (family).

The same applies to a recent result by Amiri and Wiederhake~\cite{amiri22distributed}, who consider the more general setting of $k$-hop domination, but restrict the input graphs to not only have bounded $k$-hop expansion, but also girth at least $4k+3$; the girth of a graph is the length of a shortest cycle.
They show how to achieve an approximation ratio of $O(k f(k))$ in $3k$ rounds.

In this work, we provide a number of results for $k$-hop domination on such graphs that are tight in various ways.
Most prominently, for any fixed $k$, we provide a constant-time algorithm with approximation ratio \emph{sublinear} in $f(k)$;
surprisingly, a matching lower bound shows that there is no deterministic constant-time approximation independent of $f(k)$.

\paragraph*{Detailed Contribution.}
All our results are presented in the port numbering model.
Concretely, the network is represented by a connected graph $G=(V,E)$ of $k$-hop expansion at most $f(k)$ and girth at least $4k+3$.
Each node $v \in V$ initially only knows its incident edges, locally labeled by $1,2,\ldots,\delta(v)$, where $\delta(v)$ is the degree of $v$.
Computation proceeds in \emph{synchronous} rounds, where in each round, $v\in V$ sends messages to neighbors, receives the messages of its neighbors (being aware of the port number of the connecting edge), and performs arbitrary deterministic local computations.
The running time of an algorithm is the maximum number of rounds until all nodes terminate and output whether they belong to the selected $k$-hop dominating set.

Our algorithms use messages of size at most $\lceil\log \Delta\rceil$, where $\Delta$ is the maximum degree of a node.
Our lower bounds extend to the Local Model, in which nodes have unique identifiers and messages unbounded size, which follows from known results~\cite{goeoes13lower}.
However, they do not hold when randomization is available.

As a warm-up result, we provide a simpler and faster algorithm for achieving the approximation ratio of $O(k f(k))$ given in~\cite{amiri22distributed}.
It runs for precisely $k$ rounds and sends only empty messages.
In each round, the algorithm logically deletes nodes of degree $1$ (checking for the special case that all nodes are deleted);
all nodes that remain are in the dominating set.
The fact that this procedure results in approximation ratio $O(kf(k))$ is implicit in Wiederhake's thesis~\cite{wiederhake22pulse}.
\begin{theorem}\label{thm:kround}
\Cref{alg:kround} runs for $k$ rounds, sends at most one empty message in each direction over each edge, and returns a $k$-hop dominating set that is at most by factor $2kf(k)+1$ larger than the optimum.
\end{theorem}
By considering a selection of trees of depth $k$, we establish that $k$ rounds are necessary to achieve an approximation ratio that is independent of the number of nodes $n$ and the maximum degree $\Delta$;
trees satisfy that $f(k)<1$ for any $k$.
\begin{theorem}\label{thm:kroundlower}
No $(k-1)$-round algorithm can an approximation ratio smaller than $\Delta$, even if the input graph is guaranteed to be a tree.
\end{theorem}
We remark that the fact that this lower bound is shown by invoking a family of trees is no coincidence.
In fact, we show that the task becomes easier when excluding trees, as this eliminates the special case that the above simple procedure logically deletes all nodes.
\begin{theorem}\label{thm:k-1}
\Cref{alg:k-1} runs for $k-1$ rounds and sends at most one empty message over each edge.
If the input graph is further constrained to not be a tree, it returns a $k$-hop dominating set that is at most by factor $2kf(k)+1$ larger than the optimum.
\end{theorem}
However, the improvement in running time is limited to a single round.
\begin{theorem}\label{thm:k-1lower}
No $(k-2)$-round algorithm can achieve an approximation ratio smaller than $\Delta$, even when $G$ is guaranteed to not be a tree, have maximum degree $\Delta\in \mathbb{N}$, $k$-hop expansion $f(k) \le 1$, and girth at least $g \le n(\Delta-1)/\Delta^{k+1}$.
\end{theorem}

Our main results are matching bounds on the approximation ratio that can be achieved by constant-round algorithms.
The upper bound is obtained by an algorithm that is similar to that of Amiri and Wiederhake.
We also use their key argument, that the optimum solution induces a Voronoi partition into trees, whose outgoing edges must connect to different cells and hence fully contribute to $f(k)$.
However, new ideas are required to obtain the better approximation ratio.
We exploit that the number of high-degree nodes that can be selected obeys a smaller bound due to the fact that they contribute (directly or indirectly) many edges leaving the Voronoi cell of their dominator.
On the other hand, the number of selected low-degree nodes is limited due to a preference for choosing high-degree dominators.
As the distinction between ``low'' and ``high'' is only needed in the analysis, our algorithm remains agnostic of $f(k)$.
\begin{theorem}\label{thm:3k}
Running \Cref{alg:preprocess} and then \Cref{alg:select} requires $3k$ rounds and messages of size $\lceil\log \Delta\rceil$, and returns an $O(k+f(k)^{k/(k+1)})$-approximate $k$-hop MDS.
\end{theorem}
Despite the similarities between the algorithms, our changes are necessary to achieve the stronger bound: in his thesis~\cite{wiederhake22pulse}, Wiederhake constructs a graph on which the approximation ratio of the algorithm from~\cite{amiri22distributed} is at least $kf(k)$.

Our matching unconditional lower bound is derived from a graph of large girth $g$ and uniform degree $\Delta$, where $\Delta$ is the largest integer satisfying $\Delta(\Delta-1)^k/2\leq f(k)$.
By ensuring that the graph is also bipartite, we obtain a $\Delta$-coloring of the edges.
Thus, we can assign port numbers by edge color, resulting in identical views of all nodes, regardless of running time.
This forces any port numbering algorithm to select all nodes.

Graphs of uniform degree and large girth are known to exist~\cite{erdoes63abschaetzung}, cf.~\cite{coupette20breezing} for an English version we base our proof on.
However, in addition we need to make sure that the graph actually has a small $k$-hop dominating set.
To this end, we slightly modify Erd\"os' proof of existence.
Instead of inductively increasing the degree of a uniform graph of high girth, we ``plant'' a small $k$-hop dominating set $M$ by starting from a forest of $|M|$ trees of depth $k$ and uniform degree $\Delta$ of inner nodes.
We then follow Erd\"os' approach, but limited to the set of leaves of such trees, inductively increasing their degree to $\Delta$ while maintaining the original trees.
\begin{theorem}\label{thm:lower}
Any algorithm has approximation ratio $\Omega(f(k)^{k/(k+1)})$.
\end{theorem}
We note that an $\Omega(k)$ lower bound follows from the inability to break symmetry on rings;
see~\cite{linial92locality} for the classic $\Omega(\log^* n)$ lower bound for the setting with IDs.

\paragraph*{Further Related Work.}
Above, we exclusively discussed constant-time constant-approximation algorithms, i.e., neither running time nor approximation ratio should depend on the number of nodes $n$ or the maximum degree $\Delta$ of the graph.
Due to the sheer size of the body of work on distributed dominating set approximation, a survey would be required to do it justice.
Accordingly, we confine ourselves to a few points putting our work into context.
\begin{itemize}
  \item An equivalent definition of a $k$-hop dominating set is to ask for a dominating set of the $k$-th power of the input graph. Therefore, up to a possible factor of at most $k$ in running time, considering $k>1$ makes a difference only when message size is bounded. 
  \item Approximating minimum dominating sets better than factor $\ln \Delta$ is NP-hard~\cite{dinur14analytical}. More general ``sparse'' graph classes, e.g., graphs of bounded arboricity $\alpha$, include all graphs of maximum degree $\alpha$. Therefore, one should not expect constant approximations by simple---more precisely, computationally efficient---algorithms. However, this might be different for more restrictive graph classes, and this hardness result leaves a lot of room for sublinear approximations in terms of density parameters.
  \item As little as $O(\log^* n)$ rounds make a substantial difference when unique identifiers are available. For instance, deterministic $k$-hop domination on rings with approximation ratio independent of $k$ requires $\Theta(\log^* n)$ rounds~\cite{cole86deterministic,linial92locality}. Moreover, $\Theta(\log^* n)$ rounds are enough to turn several, if not all, of the above constant-factor approximations into approximation schemes~\cite{amiri19distributed,czygrinow08fast,czygrinow18distributed,czygrinow22distributed}.
  \item In general graphs, there is a constant $c>0$ such that in $r$ rounds, one cannot achieve approximation ratio better than $\Omega(n^{c/r^2}/r)$ or $\Omega(\Delta^{1/(r+1)}/r)$ as function of $n$ and $\Delta$, respectively~\cite{kuhn16local}. These bounds are matched up to small factors by randomized algorithms using large messages~\cite{kuhn16local}.
  \item In graphs of arboricity $\alpha$ within $O(\log \Delta)$ rounds, a deterministic $O(\alpha)$-approximation is feasible~\cite{dory22near}. A reduction from the lower bound in~\cite{kuhn16local} shows that this running time is asymptotically optimal. Note that the class of graphs of bounded arboricity contains all the above graph classes for which constant-time constant-factor approximations are known.
\end{itemize}

\section{Time-optimal Algorithms}

In this section, we characterize the minimum round complexity for achieving a constant-factor approximation, i.e., one that depends on $k$ and $f(k)$ only.
\medskip

\paragraph*{Structural Characterization and Pruning Algorithms.}

We start by specifying a set of nodes that is ``safe'' to select in the sense that it is at most factor $2kf(k)$ larger than a minimum $k$-hop dominating set.
In contrast to later material, the results in this subsection are mostly implicit in~\cite{wiederhake22pulse}.
Wiederhake proved the structural properties shown here to bound the approximation ratio of the algorithm given in~\cite{amiri22distributed}, without taking the step of presenting the resulting simpler and faster algorithms we give here.
Moreover, we provide what we consider a simplified exposition with minor additions serving our needs in later sections.
The easiest way to describe the node set we are interested in is by an extremely simple algorithm: delete nodes of degree $1$ for $k$ iterations and keep what remains, cf.~\Cref{fig:prune}.
Note that this requires only $k-1$ rounds of communication, cf. \Cref{alg:k-1}.
\begin{figure}[t!]
	\centering
	\includegraphics[page=7,width=.6\textwidth]{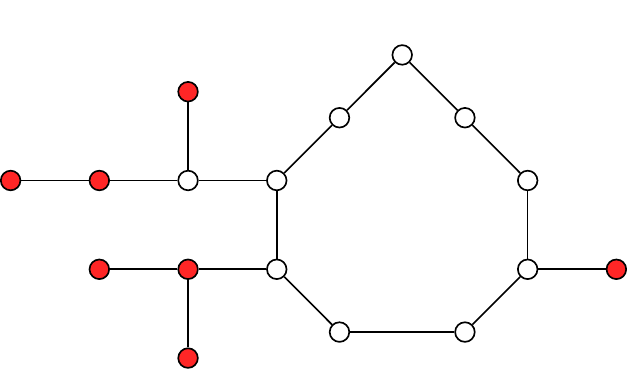}
	\caption{Example for the pruning procedure with k being 2. The deleted nodes are marked in red. Note that the set of deleted nodes induces a forest, i.e., nodes in a cycle are never deleted.}\label{fig:prune}
\end{figure}

\begin{algorithm}
\caption{Pruning algorithm, code at node v.}\label{alg:k-1}
\begin{algorithmic}[1]
\For{$k-1$ rounds}
  \If{exactly one neighbor has not sent ``del'' to $v$ yet}
    \State send ``del'' to this neighbor
  \EndIf
\EndFor
\If{sent ``del'' or exactly one neighbor has not sent ``del'' to $v$}
  \State \Return{false}
\Else
  \State \Return{true}
\EndIf
\end{algorithmic}
\end{algorithm}
Note that a node outputting $true$ deletes himself. For all presented algorithms, the return values denote the output of a node $v \in V$. For convenience, we introduce the following notation.

\begin{definition}[Pruned Graph]
Throughout this work, we denote by $G'=(V',E')$ the input graph and by $G=(V,E)$ the \emph{pruned} (sub)graph induced by the nodes that output true when running \Cref{alg:k-1}.
\end{definition}
It is possible that $G$ is empty, but this happens only if $G'$ is a tree of depth at most $k$.
\begin{observation}\label{obs:nontree}
If $G'$ is not a tree, then $V\neq \emptyset$.
\end{observation}
\begin{proof}
As $G'$ is connected and not a tree, it contains a cycle.
So long as no node from the cycle is deleted, all nodes in the cycle maintain degree at least $2$.
Hence, no node in the cycle can get deleted (first), implying that $V\neq \emptyset$.
\end{proof}
We address the special case $V=\emptyset$ later, in \Cref{alg:kround}.
Otherwise, the pruning algorithm results in a subgraph containing a $k$-hop MDS of the input graph $G'$.
\begin{lemma}\label{lem:dom}
If $V\neq \emptyset$, there is a $k$-hop MDS $M\subseteq V$ of $G'$ with the following property.
For $v\in V$, let
\begin{equation*}
r_v := \begin{cases}
  0 & \mbox{if $v$ has no neighbor in $V'\setminus V$}\\
  r & \mbox{if $r$ is the latest round in which a neighbor of $v$ sent ``del''}\\
  k & \mbox{if no neighbor of $v$ in $V'\setminus V\neq \emptyset$ sent a ``del'' message,}
\end{cases}
\end{equation*}
i.e., $r_v$ is the latest iteration of \Cref{alg:k-1} when a neigbhor of $v$ is deleted. Then there is $m\in M$ within distance $k-r_v$ of $v$; note that this distance as well as the shortest path depends on G.
Moreover, any set with these properties is a $k$-hop dominating set of $G'$.
\end{lemma}
\begin{proof}
Observe that the connected components of the subgraph induced by deleted nodes are trees.
As the graph is connected and $V\neq \emptyset$, each such tree is connected to a unique $v\in V$ by a single edge.
If for such a tree the last deleted node is deleted in round $r$, induction shows that the most distant leaf of the tree from $v$ is precisely in distance $r$ from $v$.
Thus, any $k$-hop MDS $M\subseteq V$ must contain a node within distance $k-r_v$ of $v$, and this is sufficient to cover all nodes in the deleted tree component.
The first statement of the lemma now follows from taking any $k$-hop MDS $M'$ of $G'$ and replacing any deleted dominator $m'\in M'$ by the closest $v\in V$, i.e., the unique node to which the deleted tree component of $m'$ is attached to.
As the result is a $k$-hop MDS $M\subseteq V$ of $G'$ of size equal to $M'$, it meets all requirements.

To show the second claim of the lemma, note that $k-r_v\leq k$ for all possible values of $r_v$, implying that any set with the required properties covers $V$. 
Recalling that having a dominator within distance $k-r_v$ of any $v\in V$ with $r_v>0$ is sufficient to cover all nodes in $V'\setminus V$, the second claim follows.
\end{proof}
We now show that, due to the high girth and small $k$-hop expansion of $G'$ and thus also $G$, $|V|$ is not much larger than $|M|$.
To this end, we partition $G$ with respect to $M$.
\begin{definition}[MDS Voronoi decomposition]
Let $\{T_m\}_{m\in M}$ be a Voronoi partition of $G$ with respect to the $k$-hop MDS $M$ given by \Cref{lem:dom}, i.e., $T_m$ contains all $v\in V$ for which $m$ is the closest dominating node (ties are broken arbitrarily).
For $v\in V$, denote by $m_v\in M$ the node such that $v\in T_m$.
\end{definition}
Since the girth is at least $4k+3$, for any two distinct $m,m'\in M$, there can be at most one edge between the trees the partitions $T_m$ and $T_{m'}$ induce.
On the other hand, the pruning ensures that each leaf in such a tree has a neighbor in a different tree.
\begin{lemma}\label{lem:voronoi}
For each $m\in M$, $T_m$ induces a tree of depth at most $k$ (when rooted at $m$).
Unless $T_m=\{m\}$ is a singleton, each leaf of $T_m$ has at least one edge to a node outside of $T_m$.
No two edges leaving $T_m$ connect to the same $T_{m'}$ for some $m'\in M\setminus \{m\}$.
\end{lemma}
\begin{proof}
Because for each $v\in T_m$, $m$ is the closest dominator, the subgraph induced by $T_m$ is connected.
Moreover, $m$ is within distance $k$ of all nodes in $T_m$, and any additional edge within $T_m$ beyond those of a breadth-first search tree would result in a cycle of length at most $2k+1$.
As the girth is larger than $2k+1$, $T_m$ induces a connected acyclic subgraph, i.e, a tree.

Concerning the second statement, observe that $m$ is a leaf if and only if $T_m$ is a singleton.
Hence, assume for contradiction that there is a leaf $m\neq v\in T_m$ without a neighbor in $V$.
As $v\in T_m\subseteq V$, it has not been deleted and must hence have a neighbor that was deleted in the final iteration of \Cref{alg:k-1}, i.e., $r_v=k$.
By \Cref{lem:dom}, $v$ then has a dominator in distance $k-r_v=0$, i.e., $v\in M$ and hence $v=m$, a contradiction.

The third statement again follows from the girth constraint.
Assuming for contradiction that there are $m,m'\in M$, $m\neq m'$, such that there are two edges between $T_m$ and $T_{m'}$, we obtain a cycle of length at most $4k+2$ by joining the paths of length at most $k$ from their endpoints to $m$ or $m'$, respectively, with the edges themselves.
As the girth is at least $4k+3$, such a cycle cannot exist.
\end{proof}
For ease of notation, we will use $T_m$ to refer to both the node set and its induced tree.

Together, the above properties imply that there are at most $2f(k)|M|$ leaves in the trees $T_m$, $m\in M$, in total.
Using that the depth of each $T_m$ is at most $k$, this bounds $|V|$.
\begin{lemma}\label{lem:usefk}
$|V|\le (2kf(k)+1)|M|$.
\end{lemma}
\begin{proof}
Starting from $G'$, we delete $V'\setminus V$ and contract $T_m$ for each $m\in M$.
By \Cref{lem:voronoi}, the subgraph induced by $T_m$ has radius at most $k$, i.e., the by the definition of $k$-hop expansion, the resulting minor $H$ of $G'$ has edge density at most $f(k)$.
Observe that $H$ has $|M|$ nodes.
Moreover, for each $m\in M$, unless $T_m$ is a singleton without a neighbor in $V$, by \Cref{lem:voronoi} the degree of the node resulting from contracting $T_m$ is at least the number of leaves of $T_m$.
If $T_m=\{m\}$ without a neighbor in $V$, it follows that $V=\{m\}$, since the connectivity of $G'$ implies that also $G$ is connected, too.
As in this case $|V|=|M|=1$, proceed under the assumption that this special case does not apply.
Thus, denoting by $L$ the total number of leaves in the trees $T_m$, $m\in M$, we have that the minor $H$ has at least $|L|/2$ edges.
It follows that $|L|\le 2f(k)|M|$.

To complete the proof, root for each $m\in M$ the tree induced by $T_m$ at $m$.
As by \Cref{lem:voronoi} the depth of such a tree is at most $k$, each leaf $\ell \in L$ has at most $k-1$ ancestors different from $m$.
Hence $|V|\le |L| + (k-1)|L| + |M| \le (2kf(k)+1)|M|$.
\end{proof}
\Cref{thm:k-1} readily follows from the above results.
\begin{proof}[of \Cref{thm:k-1}]
By \Cref{obs:nontree}, $V\neq \emptyset$.
By \Cref{lem:dom}, $V$ thus is a $k$-hop dominating set of $G'$.
By \Cref{lem:usefk}, $|V|\le (2kf(k)+1)|M|$.
\end{proof}
As we will discuss shortly, the exclusion of trees is indeed necessary for $k-1$ rounds to be sufficient to obtain a good approximation ratio.
However, unsurprisingly this special case can be addressed with little overhead.
Adding one more round of communication, it can be checked whether $V=\emptyset$.
If this is the case, this can be locally fixed by re-adding the (at most two) nodes that have been deleted last.
\begin{algorithm}
\begin{algorithmic}[1]
\State Execute \Cref{alg:k-1}.
\If{$v$ returned false and did not send ``del''}
  \State send ``del'' to the neigbhor from which no ``del'' was received
\EndIf
\If{$v$ returned true or both sent and received a ``del'' message in the same round}
  \State \Return{true}
\Else
  \State \Return{false}
\EndIf
\end{algorithmic}
\caption{Pruning algorithm with fallback for trees, code at v.}\label{alg:kround}
\end{algorithm}
\begin{lemma}\label{lem:del}
If $V\neq \emptyset$, \Cref{alg:kround} returns $V$.
If $V=\emptyset$, \Cref{alg:kround} returns a $k$-hop MDS of $G'$ of size at most $2$.
\end{lemma}
\begin{proof}
\Cref{alg:kround} checks whether both a node and its remaining neighbor have been deleted in the same round.
This is only possible if the subgraph induced by nodes that did not sent ``del'' messages yet consists of two nodes joined by an edge.
Thus, \Cref{alg:kround} adds nodes compared to \Cref{alg:k-1} if and only if $V=\emptyset$.
If it does add nodes, it adds exactly the two nodes that were deleted last.

It remains to prove that in this case these two nodes $m$ and $m'$ constitute a $k$-hop dominating set.
To see this, consider an arbitrary $v'\in V'$.
Since $G'$ is connected, there is a path from $v'$ to both $m$ and $m'$.
Consider the shortest path from $v'$ to one of these nodes; w.l.o.g., say this is $m'$.

Suppose $(v'=p_0,p_1,\ldots,p_{\ell-1},p_{\ell}=m')$, i.e., the path has length $\ell$.
We claim that for $i\in \{1,\ldots,\ell\}$, $p_{i-1}$ must be deleted before $p_i$.
We show this by induction on decreasing $i$.
For the base case of $i=\ell$, recall that $m'=p_{\ell}$ is deleted in the last round, and the only other node deleted in the last round is $m$.
As $m$ is not closer to $v'=p_0$ than $m'$, but $p_{\ell-1}$ is, $p_{\ell-1}$ must have been deleted before $m'=p_{\ell}$.

For the step from $i+1\in \{2,\ldots,\ell\}$ to $i$, note that by the induction hypothesis, $p_{i+1}$ is deleted after $p_i$.
Hence, $p_i$ must have received a ``del'' message from $p_{i-1}$ prior to sending one itself.
Thus, $p_{i-1}$ was deleted prior to $p_i$, completing the induction step.

Overall, it follows that $v'=p_0$ must have been deleted at least $\ell$ iterations earlier than $m'$.
However, there are only $k$ iterations, implying that $\ell<k$.
\end{proof}
Applying this lemma, we arrive at \Cref{thm:kround}.
\begin{proof}[of \Cref{thm:kround}]
If $V\neq \emptyset$, by \Cref{lem:del} the returned set is $V$.
By \Cref{lem:dom}, $V$ then is a $k$-hop dominating set of $G'$ and by \Cref{lem:usefk} the approximation guarantee is met.
If $V=\emptyset$, by \Cref{lem:del} the returned set is a $k$-hop dominating set of size $2$.
As the graph contains at least $2$ nodes and is connected, $|E'|/|V'|\ge (|V'|-1)/|V'|\ge 1/2$.
We conclude that $(2kf(k)+1)|M|\ge (2\cdot 1\cdot 1/2+1)\cdot 1=2$, proving the approximation guarantee.
\end{proof}

\paragraph*{Lower Bounds Showing Strict Time-Optimality.}
To show that the above simple pruning algorithms are indeed time-optimal for achieving a constant approximation ratio, we present matching lower bounds.
First, we show that $k$ rounds are needed to do better than a $\Delta$-approximation, even in trees of maximum degree $\Delta$.
\begin{figure}[t!]
	\centering
	\includegraphics[page=3,width=.45\textwidth]{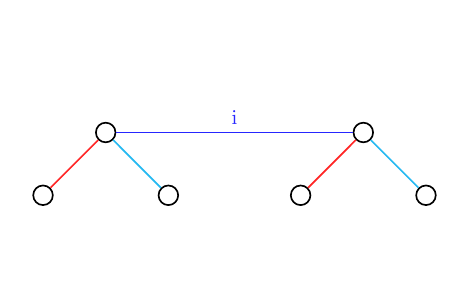}\hfill
	\includegraphics[page=4,width=.45\textwidth]{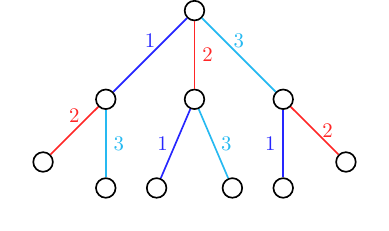}
	\caption{Graphs used in the proof of \Cref{thm:kroundlower}, for 2-hop domination and maximum degree 3. On the left, the tree consisting of two copies of T joined by an edge of color i is depicted. On the right, for each i a copy of T is attached to a root by an edge of color i. Matching port numbers to colors, nodes in the i-th copy of T cannot distinguish between the two graphs in fewer than k rounds.}\label{fig:lower_k}
\end{figure}
\begin{proof}[of \Cref{thm:kroundlower}]
Fix any port numbering algorithm that outputs a $k$-hop dominating set on input graphs from a family that contains trees of maximum degree $\Delta$.
For $i\in \{1,\ldots,\Delta\}$, denote by $T_i$ the unique properly $\Delta$-edge-colored $(\Delta-1)$-ary tree of depth $k-1$ in which the ``missing'' color among the edges incident to the root is $i$.
Consider the following properly port-numbered trees, cf.~\Cref{fig:lower_k}.
\begin{itemize}
  \item For $i\in \{1,\ldots,\Delta\}$, take two copies of $T_i$, connect their roots with an edge of color $i$, and assign port numbers to match edge colors.
  \item Take one copy of $T_i$ for each $i\in \{1,\ldots,\Delta\}$ and connect it by an edge of color $i$ to a single root node. Again, we assign port numbers to match edge colors.
\end{itemize}
Observe that for each node in each $T_i$, the port-numbered $(k-1)$-hop neighborhood in the first graph for $i$ and the second graph are isomorphic:
all non-leaves have $\Delta$ neighbors with the port numbers of each traversed edge being the same at its endpoints, and the leaves within distance $k-1$ are exactly those of $T_i$.

For each $i\in \{1,\ldots,\Delta\}$, in the first graph the algorithm must select \emph{some} node to ensure $k$-hop domination.
By indistinguishability, for each $i$ the algorithm will select the corresponding node in the copy of $T_i$ in the second graph as well.
Hence, at least $\Delta$ nodes are selected in the second graph, yet the root node $k$-hop dominates the entire graph.
\end{proof}
In other words, even infinite girth and $f(k)=1$ is not good enough for a constant approximation within $k-1$ rounds.
In~\cite{goeoes13lower}, it is shown how to lift this result to the local model.
\begin{corollary}[of Theorem~1.3 in~\cite{goeoes13lower}]\label{cor:kroundlower}
\Cref{thm:kroundlower} applies also when nodes have unique IDs.
\end{corollary}
\begin{proof}[of \Cref{cor:kroundlower}]
Assume towards a contradiction that there is local algorithm (i.e., one making use of unique IDs) that achieves constant approximation ratio $r(k,f(k))$ on trees of maximum degree $r(k,f(k))+1$.
Then this algorithm achieves the same approximation ratio on forests, as a $k$-hop MDS of a forest is a disjoint union of $k$-hop MDS of its constituent trees, and no communication is possible between different trees. 
As the family of forests of maximum degree $r(k,f(k))+1$ is closed under lifts, Theorem~1.3 in~\cite{goeoes13lower} implies that there is a port numbering algorithm achieving approximation ratio $r(k,f(k))$ on forests of maximum degree $r(k,f(k))+1$.
This is a contradiction to \Cref{thm:kroundlower}, which implies that this algorithm must have approximation ratio at least $r(k,f(k))+1$.
\end{proof}
At first glance, it might seem odd that the lower bound in the port numbering model solely relies on trees.
However, as demonstrated by \Cref{thm:k-1}, excluding trees enables us to compute a constant-factor approximation within $k-1$ rounds.
As it turns out, this improvement is limited to precisely one round.
\begin{proof}[of \Cref{thm:k-1lower}]
The special case $\Delta=2$ is addressed by considering cycles of $n\in (2k+1)\mathbb{N}$ nodes with alternating port numbers, in which a port numbering algorithm selects all nodes, but a $k$-hop MDS has size $n/(2k+1)$.
Hence, suppose that $\Delta\ge 3$ in the following.

\begin{figure}[ht!]
	\centering
	\includegraphics[page=5,width=.4\textwidth]{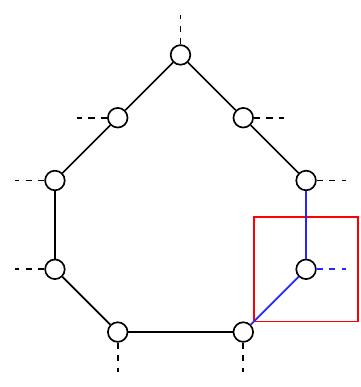}\hfill
	\includegraphics[page=6,width=.5\textwidth]{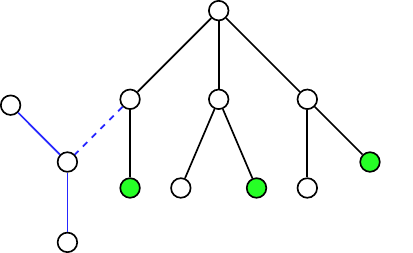}
	\caption{Lower bound graph of \Cref{thm:k-1lower}, for 2-hop domination and maximum degree 3. On the left, the underlying ring and the sole cycle in the graph is depicted. On the right for one of the nodes the attached tree is depicted, with the root at the top. If the algorithm does not select the root and its children, it must cover each of the nodes marked in green by a different dominator.}\label{fig:lower_kminus1}
\end{figure}

Denote by $T$ a complete rooted tree of depth $k$ and uniform degree $\Delta\ge 3$ at all inner nodes.
Properly edge-color $T$ with $\Delta$ colors, giving rise to a port numbering in which each port number is given by the color of the corresponding edge.
Now take a cycle $C$ of length $g$, replace each of its nodes by a copy of $T$, and choose an arbitrary leaf in each such copy.
For each edge in $C$, connect the chosen leaves of the trees corresponding to the endpoints of the edge.
We make the following observations about the resulting graph $H$, cf.~\Cref{fig:lower_kminus1}.
\begin{itemize}
  \item $H$ is a pseudo-forest with one cycle of length $g$, consisting of the edges corresponding to those of $C$.
  \item $H$ has a $k$-hop dominating set of $g$ nodes, consisting of the roots of the trees.
  \item $H$ has $f(k)=1$: pseudo-forests are closed under taking minors and have no more edges than nodes.
  \item The $(k-2)$-hop neighborhoods of the tree roots and their children are all isomorphic to one another, as any (former) leaf is in distance $k$ from the root.
\end{itemize}
It follows that any port numbering algorithm running in $k-2$ or fewer rounds must either select all roots and their children or none of these nodes.
The former results in approximation ratio at least $\Delta+1$.

We claim that the latter implies approximation ratio at least $\Delta$.
To see this, select in each of the above subtrees rooted at the child of a root one leaf, where for the subtree with the leaf receiving additional cycle edges we maximize the distance to this leaf.
Observe that
\begin{itemize}
  \item any pair of such leafs in the same tree is in distance $2k$, with the only node within distance $k$ of both of them being the root, and
  \item any pair of such leafs in different trees is in distance at least $2\cdot 2(k-1)+1>2k$; here we used that $k\ge 2$ or trivially no $(k-2)$-round algorithms exist.
\end{itemize}
As the roots are not selected, this entails that no selected node $k$-hop dominates more than one of these leaves.
The claimed lower bound of $\Delta$ also follows in this case.
Since $H$ satisfies all constraints imposed by the theorem (girth $g=n/|T|\ge n(\Delta-1)/\Delta^k$, $f(k)\le 1$, maximum degree $\Delta$, and containing a cycle), this completes the proof.
\end{proof}

\begin{corollary}[of Theorem~1.4 in~\cite{goeoes13lower}]\label{cor:k-1lower}
\Cref{thm:k-1lower} applies also when nodes have unique IDs.
\end{corollary}
\begin{proof}
The family of connected pseudoforests is closed under connected lifts.
To see this, consider any connected pseudoforest $P$ and its connected lift.
If $P$ is a tree, then its connected lift is also a tree.

Hence, assume that $P$ contains a unique cycle.
We iteratively delete nodes of degree $1$ in $P$ and their pre-images under the lift.
Note that lifts preserve degrees, also when considering subgraphs and induced subgraphs of their pre-images under the lift.
Hence, each deleted node in the pre-image also has degree $1$ at the time of deletion.
Therefore, both graphs remain connected during the process and no cycles are removed.
As $P$ contains a unique cycle, the process stops when only the nodes of this cycle and their pre-images under the lift remain.
Using again that degrees in subgraphs are preserved under the lift, all remaining nodes, including those in the lift, have degree $2$.
As we maintained connectivity also in the lift, the pre-image of the cycle of $P$ is also the unique cycle in the lift.
Therefore, the lift of $P$ is also a connected pseudoforest.

As the graph $H$ constructed in the proof of \Cref{thm:k-1lower} is a connected pseudoforest, we conclude that we can apply Theorem~1.4 in~\cite{goeoes13lower} to prove the claim.
\end{proof}

\section{Tight Approximation}
In this section, we show that the approximation ratio that can be achieved in constant time is $\Theta(k+f(k)^{k/(k+1)})$.
In fact, it turns out that in the port numbering model no algorithm can do better.
In the graph we construct for the lower bound, all views are identical, regardless of the number of rounds of communication.
Hence, unique identifiers or randomization are required to achieve better approximation guarantees.
Several proofs are deferred to the appendix.

\paragraph*{Upper Bound.}
To improve over \Cref{alg:kround}, we first apply the same pruning procedure.
We then select dominators from the candidate set given by the nodes returning true.
To this end, each node $v\in V$ with at least one deleted neighbor takes note of the maximum distance $r_v$ of a node that needs to be covered via its deleted incident edges;
we set $r_v:=0$ for nodes $v\in V$ without neighbor in $V'\setminus V$.
Observe that the components of the subgraph induced by deleted nodes are trees, each of which is connected by a single deleted edge to a non-deleted node.
Hence, there is a minimum $k$-hop dominating set without deleted nodes.
Moreover, for any $k$-hop dominating set that contains no deleted nodes, it is necessary and sufficient that it contains for each $v\in V$ a dominator within distance $k-r_v$.
If for $v\in V$ it holds that $r_v\neq 0$, it simply equals the latest round in which a neighbor is deleted, cf.~\Cref{lem:dom}.
A simple modification of \Cref{alg:kround} outputs $r_v$ for each $v\in V$, see~\Cref{alg:preprocess}. 

\begin{algorithm}
\begin{algorithmic}[1]
\State $r_v:=0$
\For{rounds $r\in \{1,\ldots,k\}$}
  \If{exactly one neighbor has not sent ``del'' to $v$ yet}
    \State send ``del'' to this neighbor
  \EndIf
  \If{received ``del''}
    \State $r_v:=r$
  \EndIf
\EndFor
\If{sent ``del'' and did not receive ``del'' the same round}
  \State \Return{false}
\Else
  \State \Return{$r_v$}
\EndIf
\end{algorithmic}
\caption{Distance-aware pruning, code at v.}\label{alg:preprocess}
\end{algorithm}
\begin{lemma}\label{lem:preprocess}
When executing \Cref{alg:preprocess}, $v\in V$ returns $r_v\in \{0,\ldots,k\}$ if and only if it returns true when \Cref{alg:kround} is run instead.
Moreover, if $V\neq \emptyset$, this return value is equal to the value $r_v$ as defined in \Cref{lem:dom}. 
\end{lemma}
\begin{proof}
Compared to \Cref{alg:k-1}, \Cref{alg:kround} adds a $k$-th round of communication in which ``del'' messages are sent by the nodes which were deleted in the $k$-th iteration.
Hence, the inclusion criterion for the set of nodes not returning false, i.e., either not being deleted or being deleted in the same round as the unique neighbor that was not yet deleted in that round, is identical for both algorithms.

\Cref{alg:preprocess} also takes note of the latest round in which a neighbor is deleted by updating $r_v$ whenever a neighbor is deleted.
Provided that $V$ (the nodes that output true when running \Cref{alg:k-1}) is identical to the set of nodes that return non-false values, this matches the definition of $r_v$ in \Cref{lem:dom}: this holds trivially true for $r_v\neq k$ and thus must extend to the remaining case of $r_v=k$.
By \Cref{lem:del}, $V\neq \emptyset$ implies that this is indeed the case.
\end{proof}
\begin{corollary}\label{cor:preprocess}
If $V=\emptyset$, the nodes not returning false when executing \Cref{alg:preprocess} constitute a $k$-hop dominating set of size $2$.
If $V\neq \emptyset$, any $D\subseteq V$ containing for each $v\in V$ some node in distance at most $k-r_v$ of $v$ is a $k$-hop dominating set.
Furthermore, there is a minimum $k$-hop dominating set $M\subseteq V$ of $G'$ with this property.
\end{corollary}
\begin{proof}
Follows from \Cref{lem:preprocess,lem:del,lem:dom}.
\end{proof}

Except for the special case that $V=\emptyset$, in its second phase our algorithm selects a dominator for each $v\in V$.
We choose a dominator of maximum degree (with respect to the subgraph induced by non-deleted nodes) within distance $k-r_v$, where ties are broken by distance and port numberings.
Intuitively, this choice is good, because together with the high girth constraint selecting many high-degree nodes implies a high expansion.
On the other hand, low-degree nodes can only be chosen by nodes for which all ancestors in $T_{m_v}$, the tree of the MDS Voronoi decomposition they participate in, also have low degree.
Together, for any fixed $k$, this results in a bound on the approximation ratio that is sublinear in $f(k)$.

Using suitable pipelining, the implementation of the selection procedure requires only $2k$ rounds and messages of size $\lceil \log \Delta\rceil$.
Broadcasting the maximum known degree in the subgraph induced by non-deleted nodes for $k$ rounds lets each node learn the highest degree node within each distance $d\in \{0,\ldots,k\}$.
Breaking ties by smallest round number and port number for which the respective value was received in the given round over the respective edge, nodes can correctly route selection messages in a pipelined fashion.
The pseudocode for the second step is given in \Cref{alg:select}.
\begin{algorithm}
\begin{algorithmic}[1]
\State $\Delta_0:=|\{p\in \{1,\ldots,\delta(v)\}\,|\,v \mbox{ did not receive ``del'' on port }p\}|$
\For{rounds $r\in \{1,\ldots,k\}$}
  \State send $\Delta_{r-1}$ to all neighbors
  \State set $\Delta_r$ to the maximum of $\Delta_{r-1}$ and all received values and store the smallest respective port number ($0$ if $\Delta_r=\Delta_{r-1}$).
\EndFor
\State set $\Delta(v):=\Delta_{k-r_v}$\algorithmiccomment{here $r_v$ computed by \Cref{alg:preprocess} is used}
\For{rounds $r\in \{k+1,\ldots,2k\}$}
  \If{$\bot\neq \Delta(v)>\Delta_{2k-r}$}
    \State send $\Delta(v)$ to the port stored for $\Delta_{2k-r+1}$ \algorithmiccomment{port is $0$ if $\Delta_{2k-r+1}=\Delta_{2k-r}$}
    \State set $\Delta(v):=\bot$
  \EndIf
  \If{any values are received}
    \State set $\Delta(v)$ to the minimum received value \algorithmiccomment{includes own value if port was $0$}
  \EndIf
\EndFor
\If{$\Delta(v)\neq \bot$}
  \State \Return{true}
\Else
  \State \Return{false}
\EndIf
\end{algorithmic}
\caption{Selection procedure, code at v. Only nodes which did not output false after running \Cref{alg:preprocess} participate. Their output and the information on which ports a ``del'' message was received serves as input. For notational convenience, nodes may fictively send messages ``to themselves'' using port number 0.}\label{alg:select}
\end{algorithm}

We first point out that \Cref{alg:select} correctly handles the special case that $V=\emptyset$.
\begin{corollary}\label{cor:select}
If $V=\emptyset$, \Cref{alg:select} selects two nodes, which form a $k$-hop dominating set.
\end{corollary}
\begin{proof}
By \Cref{cor:preprocess}, \Cref{alg:preprocess} will delete all but two nodes, which constitute a $k$-hop dominating set.
When executing \Cref{alg:select}, these both will find that $\Delta_0=\Delta_1=\ldots=\Delta_r$, maintain $\Delta(v)$ without sending any messages, and return true.
\end{proof}
Hence, in the following assume that $V\neq \emptyset$ in all statements.
\begin{observation}\label{obs:select}
For each $v\in V$, the local variables $\Delta_r$, $r\in \{0,\ldots,k\}$, store the maximum degree within $r$ hops (with respect to $G$).
\end{observation}
\begin{proof}
As $V\neq \emptyset$, the nodes that execute \Cref{alg:select} are exactly those that did not send ``del'' messages.
Therefore, $\Delta_0$ is set to the degree in the subgraph induced by $V$, i.e., in $G$.
Induction on the round number implies that for $r\in \{1,\ldots,k\}$ the variable $\Delta_r$ is set to the maximum degree within $r$ hops, which is increasing in $r$.
\end{proof}

\begin{lemma}\label{lem:select_dominates}
The set of nodes returning true when executing \Cref{alg:select} after \Cref{alg:preprocess} is a $k$-hop dominating set of $G'$.
For each selected node $v\in V$, there is some $w\in V$ within distance $k-r_w$ of $v$ such that $\Delta_{k-r_w}(w)=\Delta_0(v)$. 
\end{lemma}
\begin{proof}
As $V\neq \emptyset$, by \Cref{cor:preprocess} it is sufficient to select for each $v\in V$ some node within distance $k-r_v$.
Therefore, to prove the claims of the lemma, associate for each $v\in V$ a token of value $\Delta_{k-r_v}$ with $v$ that is moved to the recipient when the node holding the token sends a message.
When two or more tokens meet, i.e., move to or stay at the same node, all but one are removed.
The surviving one is an arbitrary token of minimal value.

Observe that this process exhibits the following properties:
\begin{itemize}
  \item Each node initially holds a token.
  \item Because $\Delta_r$ is increasing in $r$ and tokens get removed only by tokens of smaller value, node $v\in V$ does not send a token before round $k+1-r_v$.
  \item If a token is sent, some neighbor will hold a token at the end of a round.
  \item Together, this implies that at the end of round $2k$, there is a token within distance $2k-k-r_v=k-r_v$ of each node $v\in V$.
  \item A node has $\Delta(v)\neq \bot$ (and hence outputs true) if and only if it has a token.
\end{itemize}
This proves that the output is a $k$-hop dominating set of $G'$.

To prove the second claim of the lemma, trace back the token of each selected node to its origin.
When a token is sent by $v\in V$, it is sent to the neighbor $w\in V$ that caused $\Delta_{2k-r+1}(v)$ to exceed $\Delta_{2k-r}(v)$, by a message with value $\Delta_{2k-r+1}(v)=\Delta_{2k-r}(w)$.
Applying the monotonicity of $\Delta_r$ once more, by induction we conclude that for a non-removed token originating at $v\in V$, the node $w\in V$ holding it in the end satisfies that $\Delta_0(w)=\Delta_{k-r_v}(v)$.
As the token did not move before round $k+1-r_v$, the second claim of the lemma follows. 
\end{proof}

It remains to prove the approximation guarantee.
Denote by $D\subseteq V$ the set of nodes returning true when executing \Cref{alg:select} after \Cref{alg:preprocess}.
For $d\in D$, arbitrarily fix a node $s_d\in V$ that ``selected'' it, i.e., that satisfies that $d$ has maximal degree among all nodes within distance $k-r_{s_d}$ of $s_d$;
such a node exists by \Cref{obs:select} and \Cref{lem:select_dominates}.
Let $S:=\bigcup_{d\in D}\{s_d\}$ and $\Delta\in \mathbb{N}$ be a degree threshold that we will fix later in the analysis.
Abbreviate $D_{\Delta}:=\{d\in D\,|\,\delta(d)\le \Delta\}$ and $S_{\Delta}:=\{s\in S\,|\,s=s_d\mbox{ for some }d\in D_{\Delta}\}$.
Fix a $k$-hop MDS $M$ such that for each $v\in V$, its dominator $m_v$ is in distance at most $k-r_v$;
such an $M$ exists by \Cref{cor:preprocess}.

We account for low-degree nodes in $D$ by bounding the available number of nodes that might select them.
\begin{lemma}\label{lem:low}
\begin{equation*}
|D_{\Delta}|\le
\begin{cases}
(2k+1)|M| &\mbox{if }\Delta \le 2\\
\frac{\Delta}{\Delta-2}\cdot(\Delta-1)^k |M| & \mbox{else.}
\end{cases}
\end{equation*}
\end{lemma}
\begin{proof}
Observe that $|D_{\Delta}|=|S_{\Delta}|$ by construction, so it is sufficient to bound $|S_{\Delta}|$.
By definition of $D_{\Delta}$ and $S_{\Delta}$, all nodes within distance $k-r_s$ of $s\in S_{\Delta}$ have degree at most $\Delta$.
In particular, this applies to $s$ and its ancestors in $T_{m_s}$, i.e., in the tree induced by the Voronoi cell of its dominator, because $m_s$ is within distance $k-r_s$ of $s$.
Thus, each $s\in S_{\Delta}$ is contained in the connected component of $m_s$ in $T_{m_s}$ induced the nodes of degree at most $\Delta$.

We conclude that $|S_{\Delta}|$ is bounded by $|M|$ times the maximum size of a tree of depth $k$ whose inner nodes have uniform degree $\Delta$.
For $\Delta\le 2$, this number is at most $2k+1$, as the tree is then a path.
For $\Delta\ge 3$, we get that
\begin{equation*}
\frac{|S_{\Delta}|}{|M|}
\le 1+\Delta \sum_{i=1}^k (\Delta-1)^{i-1}
\le \Delta (\Delta-1)^{k-1}\sum_{i=0}^{\infty} (\Delta-1)^{-i}
=\frac{\Delta}{\Delta-2}\cdot(\Delta-1)^k.
\end{equation*}
\end{proof}

High-degree nodes in $D$ imply a large number of edges leaving their dominator's tree.
This contributes to increasing the $k$-hop expansion of the graph, i.e., we can upper bound the number of such nodes by means of the $k$-hop expansion $f(k)$.
\begin{lemma}\label{lem:high}
If $\Delta\ge 2$, then $|D\setminus D_{\Delta}|\le 4f(k)|M|/(\Delta-1)$.
\end{lemma}
\begin{proof}
Let $D'\subseteq D\setminus D_{\Delta}$ the set of nodes $d\in D\setminus D_{\Delta}$ satisfying that there is at most one descendant $d'\in T_{m_d}$ that is also in $D\setminus D_{\Delta}$.
For such a $d\in D'$, let $T_d$ be the subtree rooted at $d$ after removing the subtree rooted at $d'$ (if there is such a descendant $d'$).
Denote by $E_d$ the set of edges that connect a node in $T_d$ to a node outside $T_m$.
We attribute all edges that leave $T_m$ from the subtree rooted at $d$ after removing the subtree rooted at the child of $d$ containing a descendant in $T_m\cap (D\setminus D_{\Delta})$ (if there is one).
Note that by \Cref{lem:voronoi}, each leaf of $T_m$ has at least one incident edge leaving $T_m$ (unless this leaf is $m$ and has no incident edge, which implies that $V=M=\{m\}$).
As $d$ has degree at least $\Delta+1$, this implies that there are at least $\Delta-1$ such edges;
apart from ``losing'' up to one edge due to $d'$, one connects to the parent of $d$.

Observe that each edge attributed to $d$ can be attributed to at most one other node, an ancestor of the edge's endpoint in $T_{m'}\in D\setminus D_{\Delta}$, where $m'\neq m$ is the dominator of the endpoint of the edge.
With this in mind, we contract $T_m$ for all $m\in M$.
As the depth of each $T_m$ is at most $k$, the resulting graph has at most $f(k)|M|$ edges.
Invoking \Cref{lem:voronoi} once more, none of the edges between trees are removed by the contractions, as each of them connects trees $T_m$ and $T_{m'}$ for a unique pair $m,m'\in M$.

Therefore, after the contraction the sum of degrees is at least
\begin{equation*}
\sum_{d\in D'}|E_d|\ge \sum_{d\in D'}\Delta-1 = (\Delta-1)|D'|.
\end{equation*}
Applying the upper bound on the number of edges, we conclude that $|D'|\le 2f(k)|M|/(\Delta-1)$.

Finally, it remains to bound $|D\setminus (D_{\Delta}\cup D')|$.
To this end, observe that by definition of $D'$, each $d\in D\setminus (D_{\Delta}\cup D')$ has at least $2$ descendants in $D\setminus D_{\Delta}$.
Thus, contracting in each $T_m$ edges with at most one endpoint in $D\setminus D_{\Delta}$ until this process stops, we end up with a forest with the following properties.
\begin{itemize}
  \item The node set can be mapped one-on-one to $D\setminus D_{\Delta}$.
  \item Nodes corresponding to those of $D'$ have degree at most $2$.
  \item Nodes corresponding to those of $D\setminus (D_{\Delta}\cup D')$ have degree at least $3$.
\end{itemize}
Further removing degree-$2$ nodes by contractions thus results in a forest in which inner nodes map one-on-one to those of $D\setminus (D_{\Delta}\cup D')$ and the number of leaves is bounded from above by $|D'|$.
Therefore,
\begin{equation*}
|D\setminus D_{\Delta}|= |D'| + |D\setminus (D_{\Delta}\cup D')|<2|D'|\le \frac{4f(k)|M|}{\Delta-1}.
\end{equation*}
\end{proof}

Choosing $\Delta$ suitably, we arrive at \Cref{thm:3k}.
\begin{proof}[of \Cref{thm:3k}]
By \Cref{cor:select} and \Cref{lem:select_dominates}, running \Cref{alg:preprocess} followed by \Cref{alg:select} results in a $k$-hop dominating set.
The running time and message size are immediate from the pseudocode.
If $f(k)^{1/{k+1}}<2$, choosing $\Delta=2$ and applying \Cref{lem:low,lem:high} yields approximation ratio $O(k+f(k))=O(k+f(k)^{k/(k+1)})$.
If $f(k)^{1/(k+1)}\ge 2$, choosing $\Delta:=\lfloor f(k)^{1/(k+1)}\rfloor$ and applying \Cref{lem:low,lem:high} yields approximation ratio $O(k+f(k)/\Delta)=O(k+f(k)^{k/(k+1)})$.
\end{proof}

\paragraph*{Lower Bound}
Our lower bound is based on a graph of uniform degree $\Delta$.
We first prove a helper statement relating the $k$-hop expansion to this degree bound.
\begin{lemma}\label{lem:contract}
Any graph has $k$-hop expansion $f(k)\le \Delta(\Delta-1)^k/2$.
\end{lemma}
\begin{proof}
We upper bound the number of outgoing edges of a subgraph of radius $k$.
We claim that there is a subgraph maximizing this value such that that all outgoing edges have the endpoint in the subgraph in distance exactly $k$ from the center.
To see this, consider any subgraph not satisfying this property, choose any endpoint of an outgoing edge in distance smaller than $k$, subdivide the edge by inserting a new node, and add this node to the subgraph.
Repeating this process until this is no longer possible results in a subgraph with the same number of outgoing edges and the above property.

Next, note that such a maximizing subgraph is a tree of depth $k$ rooted at the center.
To see this, take any subgraph of radius $k$, construct a breadth-first search tree starting from its center, delete all other internal edges, and add for each of them an outgoing edge for each of its endpoints.
Since the graph is maximal, this cannot add any outgoing edges, so it must be the case that the subgraph already was a tree.

Finally, observe that any tree of depth at most $k$ and maximum degree $\Delta$ has at most $\Delta(\Delta-1)^{k-1}$ leaves.
Thus, the above properties imply that the number of outgoing edges is bounded by $\Delta(\Delta-1)^k$.
We conclude that deleting nodes and contracting subgraphs of radius at most $k$ results in a graph of maximum degree at most $\Delta(\Delta-1)^k$.
As the sum of degrees is twice the number of edges, it follows that $f(k)\le \Delta(\Delta-1)^k/2$.
\end{proof}

We proceed to constructing the lower bound graph.
\begin{lemma}\label{lem:lb_graph}
Let $\Delta,k\in \mathbb{N}$ with $\Delta\ge 3$.
Then there is a graph of uniform degree $\Delta$, girth at least $g:=4k+3$, $k$-hop expansion $f(k)\le \Delta(\Delta-1)^k/2$, and a $k$-hop dominating set of size $O(|V|\Delta/f(k))$.
\end{lemma}
\begin{proof}
Suppose that $m:=m(\Delta,k)\in 2\mathbb{N}$ is sufficiently large.
Let $G_1$ be the disjoint union of $m$ rooted trees of depth $k$, where inner nodes have degree $\Delta$.
Clearly, $G_1$ has a $k$-hop dominating set of $m$ nodes, and adding edges between leaves of the trees will not change this.
Note that each tree has $1+\sum_{i=0}^{k-1} \Delta(\Delta-1)^i \in \Theta(\Delta(\Delta-1)^{k-1})=\Omega(f(k)/\Delta)$ nodes.
Hence, $m\in O(|V|\Delta/f(k))$.

We claim that we can add edges between leaves in $G_1$ to obtain a graph $G_{\Delta}$ in which these nodes have all degree $\Delta$ and the girth is at least $g$.
By \Cref{lem:contract}, $G_{\Delta}$ then also has $k$-hop expansion $f(k)$.
As adding edges cannot remove the property that a set is a $k$-hop dominating set, proving the claim will prove the lemma.

We prove the claim inductively, by constructing $G_i$, $i\in \{2,\ldots,\Delta\}$, where $G_i$ is $G_1$ with additional edges between leaves of $G_1$ such that these nodes all have degree $i$.
The base case of $i=1$ is given by $G_1$.
For the step from $i\in \{1,\ldots,\Delta-1\}$ to $G_{i+1}$, let $G$ be a graph maximizing the number of edges among all graphs that are $G_1$ with additional edges between leaves so that (i) the degree in $G_i$ of leaves of $G_1$ is between $i$ and $i+1$ and (ii) girth is at least $g$.
By the induction hypothesis, $G_i$ is $G_1$ with additional edges and meets conditions (i) and (ii), implying that $G$ is well-defined.
We will show that the minimum degree of $G$ is $i+1$, completing the induction step and thereby the proof.

To this end, assume for contradiction that there is a node of degree $i$ in $G$.
Denote by $G'=(V',E')$ the subgraph of $G$ induced by the leaves in $G_1$.
It contains $\ell\ge 1$ nodes of degree $i-1$ and $|V'|-\ell$ nodes of degree $i$.
As the number of leaves in $G_1$ is a multiple of $m\in 2\mathbb{N}$, $|V'|$ is even.
Using that $2|E'|=(i-1)\ell+i(|V'|-\ell)=i|V'|-\ell$, we conclude that $\ell$ is also even.
In particular, $\ell\ge 2$, i.e., there are at least two nodes of degree $i$ in $G$.

Denote by $v,w\in V'$ two nodes of degree $i$ in $G$.
Denote by $N$ the set of nodes with distance at most $g-2$ in $G$ from both $v$ and $w$.
If there is $x\notin N$ of degree $i-1$, we can add an edge from $x$ to $v$ or $w$ without decreasing the girth below $g$, contradicting the maximality of $G$.
Thus, each node in $V'\setminus N$ has degree $i+1$ in $G$ and $i$ neighbors in $V'$.
On the other hand, the constraints on $G$ imply that nodes in $V'\setminus N$ have at most $i$ neighbors in $V'$.

Observe that the size of $N$ is bounded by $1+\sum_{i=1}^{g-2}\Delta(\Delta-1)^{i-1}=1+\sum_{i=1}^{4k}\Delta(\Delta-1)^{i-1}$.
Thus, as $m(\Delta,k)$ is sufficiently large, $2|N|\le |V'|$.
It follows that $i|N\cap V'|< i|N|\le i|V'\setminus N|$.
Therefore, $G$ must contain an edge $\{x,y\}$ with both endpoints $x,y\in V'\setminus N$.
However, this also leads to a contradiction to the maximality of $G$:
We can delete $\{x,y\}$ and add $\{v,x\}$ and $\{w,y\}$; there can be no short cycle involving only one of the two new edges, since the distances from $v$ to $x$ and from $w$ to $y$ are at least $g-1$, and there can be no short cycle involving both $\{v,x\}$ and $\{w,y\}$, since otherwise $G\setminus \{x,y\}$ would contain a short path from $x$ to $y$, i.e., $G$ would have too small girth.
\end{proof}
Considering the bipartite double cover, we can ensure that the constructed graph can be $\Delta$-edge colored.
\begin{lemma}\label{lem:bipartite}
There is a graph as in \Cref{lem:lb_graph} that is also balanced bipartite.
\end{lemma}
\begin{proof}
Consider the bipartite double cover $C=(V_C,E_C)$ of the graph $G=(V,E)$ from \Cref{lem:lb_graph}.
That is, $V_C:=V_1\cup V_2$, where $V_i:=\{v_i\,|\,v\in V\}$, $i=\{1,2\}$, are two distinct copies of $V$, and
$E_C:=\bigcup_{\{v,w\}\in E}\{v_1,w_2\}\cup\{v_2,w_1\}$.
Clearly, $|V_1|=|V_2|$ and $C$ can be properly $2$-colored by assigning color $i$ to $V_i$.
Moreover, since $G$ has uniform degree $\Delta$, so has $C$, and by \Cref{lem:contract} this implies the desired bound on $f(k)$.
For any $k$-hop dominating set $D$, $\{d_1\,|\,d\in D\}\cup \{d_2\,|\,d\in D\}$ is a $k$-hop dominating set as well, showing that there is a small dominating set of $C$.

Finally, assume towards a contradiction that $C$ has a cycle of length $\ell<g$, say $(v_1^{(1)},v_2^{(2)},v_1^{(3)},\ldots,v_2^{(\ell)},v_1^{(1)})$.
Then $(v^{(1)},v^{(2)},\ldots,v^{(\ell)},v^{(1)})$ is a walk of length $\ell$ in $G$ that starts and ends at the same node, i.e., there is a cycle of length $\ell<g$ in $G$.
This is a contradiction, completing the proof.
\end{proof}

\begin{corollary}\label{cor:bipartite}
The graph given by \Cref{lem:bipartite} can be properly $\Delta$-edge colored.
\end{corollary}
\begin{proof}
As the graph is bipartite and has uniform degree, Hall's marriage theorem implies that it contains a perfect matching.
Assign color $1$ to these edges and delete them, resulting in a graph of uniform degree $\Delta-1$.
Inductive repetition yields the desired coloring.
\end{proof}
\Cref{thm:lower} now readily follows, as choosing port numbers in a graph of uniform degree $\Delta$ according to a $\Delta$-edge coloring results in identical views at all nodes.
\begin{proof}[of \Cref{thm:lower}]
Consider the graph given by \Cref{lem:bipartite} with port numbers matching the colors of a proper $\Delta$-edge coloring, which is feasible by \Cref{cor:bipartite}.
In this graph, all nodes have identical views, regardless of the number of rounds of computation.
In more detail, initially all nodes have identical state.
Thus, on each port $i\in \{1,\ldots,\Delta\}$, each node sends the same message, meaning that each node receives this message on its port $i$.
By induction, this implies that for each round of computation, all nodes end up in the same state.
Hence, any port numbering algorithm must select all nodes into a $k$-hop dominating set, yet only $O(|V|\Delta/f(k))$ many are required.
As the graph of \Cref{lem:bipartite} satisfies that $f(k)=\Delta(\Delta-1)^k/2$, for any choice of $\Delta\ge 2$ the approximation ratio is $\Omega((\Delta-1)^k)=\Omega(f(k)^{k/(k+1)})$.
\end{proof}
Applying~\cite{goeoes13lower}, this result extends to constant-round algorithms in the Local model.
\begin{corollary}
Every algorithm with running time depending on $k$ and $f(k)$ has approximation ratio $\Omega(f(k)^{k/(k+1)})$, even if nodes have unique identifiers.
\end{corollary}

\bibliographystyle{plainurl}
\bibliography{references}

\end{document}